\documentclass[11pt]{article}

\usepackage{authblk}
\author[]{Ingvar Ziemann}

\affil[]{University of Pennsylvania}

\usepackage[utf8]{inputenc} 
\usepackage[T1]{fontenc} 

\usepackage{fullpage}

\usepackage{url}            
\usepackage{booktabs}       
\usepackage{amsfonts}       
\usepackage{nicefrac}       
\usepackage{microtype}      
\usepackage[dvipsnames]{xcolor}

\usepackage[numbers]{natbib}
\usepackage{enumitem}
\usepackage{amsthm}
\usepackage{amssymb}
\usepackage{amsmath}
\usepackage{mathrsfs}
\usepackage{hyperref}
\hypersetup{
    pdftoolbar=true,        
    pdfmenubar=true,        
    pdffitwindow=false,     
    pdfstartview={FitH},    
    pdfnewwindow=true,      
    colorlinks=true,       
    linkcolor=blue,          
    citecolor=ForestGreen,        
    filecolor=magenta,      
    urlcolor=red,           
    breaklinks=false,
}

\usepackage{cleveref}
\usepackage{thmtools}
\usepackage{thm-restate}

\numberwithin{equation}{section}

\date{}

\newcommand{\e}{\varepsilon}

\newcommand{\E}{\mathbf{E}}

\DeclareMathOperator{\blkdiag}{blkdiag}

\DeclareMathOperator{\tr}{tr}

\newcommand{\T}{\mathsf{T}}

\newcommand{\R}{\mathbb{R}}
\newcommand{\N}{\mathbb{N}}
\renewcommand{\Pr}{\mathbf{P}}

\newcommand{\paren}[1]{\ensuremath{\left( #1\right)}}

\newcommand{\bignorm}[1]{\left\lVert #1 \right\rVert}

\newcommand{\opnorm}[1]{\| #1 \|_{\mathsf{op}}}
\newcommand{\bigopnorm}[1]{\left\| #1 \right\|_{\mathsf{op}}}

\newtheorem{theorem}{Theorem}[section] 
\newtheorem{corollary}{Corollary}[section] 

\newtheorem{lemma}{Lemma}[section]

\title{A note on the smallest eigenvalue of the empirical covariance of causal Gaussian processes}
\begin{document}

\maketitle

\begin{abstract}
    We present a simple proof for bounding the smallest eigenvalue of the empirical covariance in a causal Gaussian process. Along the way, we establish a one-sided tail inequality for Gaussian quadratic forms using a causal decomposition. Our proof only uses elementary facts about the Gaussian distribution and the union bound. We conclude with an example in which we provide a performance guarantee for least squares identification of a vector autoregression.
\end{abstract}


\section{Introduction}
We consider a causal Gaussian process $X_{0:T-1}=(X_0^\T,\dots,X_{T-1}^\T)^\T$ evolving on $\R^d$.
In this note we provide an elementary proof of the fact that the empirical covariance:
\begin{align}\label{eq:empcov}
    \widehat \Sigma_X \triangleq \frac{1}{T}\sum_{t=0}^{T-1} X_t X_t^\T
\end{align}
is never much smaller than its (conditional) expectation. Analyzing the lower tail of \eqref{eq:empcov} has been the subject of a number of recent papers as it is crucial to characterize the rate of convergence in linear system identification \citep{simchowitz2018learning,faradonbeh2018finite, sarkar2019near, tsiamis2019finite, oymak2021revisiting, tu2022learning, jedra2022finite, tsiamis2022statistical}. In these references, a number of  elegant but rather advanced techniques can be found to control the lower tail of \eqref{eq:empcov} for various models (but mainly linear dynamical systems). Of these, the perhaps most well-known being the adaptation of the small-ball method of \cite{mendelson2014learning} by \cite{simchowitz2018learning}. Our aim with this note is to give a more  accessible proof of these results in the Gaussian setup, but which also easily extends to any causal Gaussian process (\Cref{thm:anticonc}), e.g., ARMA processes (\Cref{sec:arma}). The main idea here is based on \cite{ziemann2022learning}, which shows that one can often encode such "small-ball behavior", even for highly dependent processes, by a one-sided exponential inequality (\Cref{thm:expineq}).

\paragraph{Motivation}
The primary reason for our interest (and that of the above-mentioned references) in \eqref{eq:empcov} is the fact that in the linear regression model:
\begin{align*} 
    Y_t &= A_\star X_t + V_t &&t=0,\dots,T-1 &&(V_t\textnormal{ noise})
\end{align*}
the error of the least squares estimator $\widehat A$ of the unknown parameter $A_\star$ can be expressed as:
\begin{equation}\label{eq:LSEerror}
\begin{aligned}
    &
    \widehat A-A_\star
     \\ &
    = \left[\left(\sum_{t=0}^{T-1} V_t X_t^\T \right)\left(\sum_{t=0}^{T-1} X_t X_t^\T \right)^{-1/2 }\right]\left(\sum_{t=0}^{T-1} X_t X_t^\T \right)^{-1/2 }.
\end{aligned}
\end{equation}
The leftmost term of \eqref{eq:LSEerror} (in square brackets) can be shown to be (almost) time-scale invariant in many situations. For instance, if the noise $V_{0:T-1}$ is a sub-Gaussian martingale difference sequence with respect to the filtration generated by the covariates $X_{0:T-1}$, one  can invoke the so-called self-normalized martingale theorem  of \cite{pena2009self,abbasi2011improved} to show this. Whenever this is the case, the dominant term in the rate of convergence of the least squares estimator is $ \left(\sum_{t=0}^{T-1} X_t X_t^\T \right)^{-1/2 }$. Thus, providing control of the smallest eigenvalue of \eqref{eq:empcov} effectively yields control of the rate of convergence of the least squares estimator in many situations.  Put differently, the smallest eigenvalue of \eqref{eq:empcov} quantifies the notion of persistency of excitation often encountered in system identification \cite{lennart1999system, willems2005note}. We also remark that two-sided bounds are often unsatisfactory for this purpose, and will indeed become hopeless for processes that are not strictly stable. Nevertheless, a one-sided bound is often still possible.

\paragraph{Notation}
For an integer $n\in \N$, we  define the shorthand $[n] \triangleq\{1,\dots,n\}$. The Euclidean norm on $\mathbb{R}^{d}$ is denoted $\|\cdot\|$,
and the unit sphere in $\R^d$ is denoted $\mathbb{S}^{d-1}$. The identity matrix acting on $\R^d$ is denoted $I_d$.  The trace of a matrix $M \in \R^{d_1\times d_2}$ is denoted $\tr M$, its tranpose $M^\T$, and its operator norm is $\opnorm{M} \triangleq\sup_{v \in \mathbb{S}^{d_2-1}} \| Mv\|$. If a square matrix $M\in \R^{d\times d}$ further is positive semidefinite, we also write $\lambda_{\min}(M)$ for its smallest eigenvalue and $\lambda_{\max}(M)$ for its largest. For $i,j\in \N$ with $i<j$, and a sequence of vectors $x_i,\dots,x_j \in \R^d$ we set $(v_i^\T,\dots,v_j^\T)^\T \triangleq  v_{i:j}\in \R^{(j-i+1)d}$.   If $M_i, i \in [N]$ are matrices, the matrix $\blkdiag(M_1,\dots,M_n)$ denotes the block matrix with the $M_i$ on its main diagonal ordered from $M_1$ (top-left) to $M_n$ (bottom-right) and all other entries identically zero. Expectation (resp.\ probability) with respect to all the randomness of the underlying probability space is denoted by $\E$ (resp.\ $\Pr$). Finally, the shorthand $W\sim N(0,I_d)$ introduces $W$ as a mean zero Gaussian random vector in $\R^d$ with covariance matrix $I_d$.

\section{Preliminaries}
Fix two integers $T$ and $k$ such that $T/k\in \N$. We consider a ($k$)-causal Gaussian process $X_{0:T-1}=(X_0^\T,\dots,X_{T-1}^\T)^\T$ evolving on $\R^d$. More precisely, we assume the existence of a Gaussian white process evolving on $\R^p$, $W_{0:T-1}\sim N(0, I_{pT})$, and a (block-) lower triangular matrix $\mathbf{L} \in \R^{dT\times pT}$ such that $X_{0:T-1}=\mathbf{L}W_{0:T-1}$. We say that $X_{0:T-1}$ is $k$-causal if the matrix $\mathbf{L}$ has the form:
\begin{align*}
    \mathbf{L} 
    = 
    \begin{bmatrix}
    \mathbf{L}_{1,1} &0 &0&0&0\\
    \mathbf{L}_{2,1} & \mathbf{L}_{2,2} & 0  &0 &0\\
    \mathbf{L}_{3,1} & \mathbf{L}_{3,2} & \mathbf{L}_{3,3}  &0 &0\\
    \vdots & \ddots & \ddots & \ddots &\vdots\\
    \mathbf{L}_{T/k,1} &\dots & \dots & \dots&\dots \mathbf{L}_{T/k,T/k}
    \end{bmatrix}
    =
    \begin{bmatrix}
    \mathbf{L}_{1}\\
    \mathbf{L}_{2}\\
    \mathbf{L}_{3}\\
    \vdots \\
    \mathbf{L}_{T/k}
    \end{bmatrix}
\end{align*}
where each $\mathbf{L}_{ij} \in \R^{dk\times pk}, i,j \in [T/k] \triangleq \{1,2,\dots,T/k\}$. Obviously, every $1$-causal process is $k$-causal for every $k\in \N$ (for appropriate $T$). To every $k$-causal Gaussian process, we also  associate a decoupled random process $\tilde X_{0:T-1} = \mathrm{blkdiag}(\mathbf{L}_{11},\dots, \mathbf{L}_{T/k,T/k})W_{0:T-1}$. This decoupled process will effectively dictate our lower bound, and we will show under relatively mild assumptions that
\begin{align*}
    \lambda_{\min}\left (\frac{1}{T}\sum_{t=0}^{T-1} X_t X_t^\T \right)\gtrsim \lambda_{\min} \left(\frac{1}{T}\sum_{t=0}^{T-1}\E \tilde X_t \tilde X_t^\T\right)
\end{align*}
with probability that approaches $1$ at an exponential rate in the sample size $T$.

Our proof will make heavy use of the following lemma.

\begin{restatable}{lemma}{gausscondlem}\label{lem:gausscondlem}
Fix $x\in \R^n$ and let $W\sim N(0,I_m)$. For any positive semidefinite $Q \in \R^{(n+m)\times(n+m)}$ of the form $ Q=\begin{bmatrix}Q_{11}& Q_{12}\\ Q_{21} & Q_{22} \end{bmatrix}$ and any $\lambda \geq 0$ we have that:
\begin{align*}
    \E \exp \left( -\lambda \begin{bmatrix}x \\ W
    \end{bmatrix}^\T  \begin{bmatrix}Q_{11}& Q_{12}\\ Q_{21} & Q_{22} \end{bmatrix} \begin{bmatrix}x \\ W
    \end{bmatrix}\right) 
    \leq\exp \left(-\lambda \tr Q_{22} + \frac{\lambda^2}{2} \tr Q_{22}^2  \right).
\end{align*}
\end{restatable}

In principle, we will use \Cref{lem:gausscondlem} to "throw away" the inter-block correlation in $\mathbf{L}$, thereby reducing the process $X_{0:T-1}$ to $\tilde X_{0:T-1}$, which is easier to analyze.
\section{Results}

Repeated application of \Cref{lem:gausscondlem} to the process $X_{0:T-1}=\mathbf{L}W_{0:T-1}$ yields our main result.

\begin{theorem}\label{thm:expineq}
Fix an integer $k \in \N$, let $T \in N$ be divisible by $k$ and suppose $X_{0:T-1}$ is a $k$-causal Gaussian process. Fix also a matrix $\Delta\in \R^{d'\times d}$. Then for every $\lambda \geq 0$:
\begin{multline*}
    \E \exp \left(-\lambda \sum_{t=0}^{T-1}\|\Delta X_t\|_2^2 \right) 
    %
    \leq \exp \Bigg( -\lambda\sum_{j=1}^{T/k} \tr\left[ \mathbf{L}_{j,j}^\T \mathrm{blkdiag}(\Delta^\T \Delta) \mathbf{L}_{j,j}\right] 
    \\
    + \frac{\lambda^2}{2} \sum_{j=1}^{T/k} \tr\left[ \mathbf{L}_{j,j}^\T \mathrm{blkdiag}(\Delta^\T \Delta) \mathbf{L}_{j,j}\right]^2 \Bigg).
\end{multline*}
\end{theorem}
It is worth pointing out that $\sum_{j=1}^{T/k} \tr\left[ \mathbf{L}_{j,j}^\T \mathrm{blkdiag}(\Delta^\T \Delta) \mathbf{L}_{j,j}\right] = \sum_{t=0}^{T-1} \E \|\Delta \tilde X_t\|_2^2 .$ Hence \Cref{thm:expineq} effectively passes the expectation inside the exponential at the cost of working with the possibly less excited process $\tilde X_{0:T-1}$ and a quadratic correction term. Note also that the assumption that $T$ is divisible by $k$ is not particularly important. If not, let $T'$ be the largest integer such that $T'/k \in \N$ and $T'\leq T$ and apply the result with $T'$ in place of $T$.

The significance of \Cref{thm:expineq} is demonstrated by the following simple calculation. Namely, for any fixed $\Delta \in \R^{d'\times d} \setminus \{0\}$ and $\lambda \geq 0$ we have that:
\begin{equation}
\label{eq:simplechernoff}
\begin{aligned}
&\mathbf{P} \left( \sum_{t=0}^{T-1}  \|\Delta X_t\|^2 \leq \frac{1}{2} \sum_{t=0}^{T-1} \E \|\Delta\tilde  X_t\|^2 \right)\\
&\leq  \E \exp \left(  \frac{\lambda}{2} \sum_{t=0}^{T-1} \E \|\Delta\tilde X_t\|^2 -\lambda \sum_{t=0}^{T-1} \|\Delta X_t\|^2   \right) \quad (\textnormal{Chernoff})\\
&\leq  \exp \Bigg( -\frac{\lambda}{2}\sum_{j=1}^{T/k} \tr\left[ \mathbf{L}_{j,j}^\T \mathrm{blkdiag}(\Delta^\T \Delta) \mathbf{L}_{j,j}\right] \\
&+ \frac{\lambda^2}{2} \sum_{j=1}^{T/k} \tr\left[ \mathbf{L}_{j,j}^\T \mathrm{blkdiag}(\Delta^\T \Delta) \mathbf{L}_{j,j}\right]^2 \Bigg) \quad (\textnormal{\Cref{thm:expineq}})\\
&= \exp \left(-\frac{\left(\sum_{j=1}^{T/k} \tr\left[ \mathbf{L}_{j,j}^\T \mathrm{blkdiag}(\Delta^\T \Delta) \mathbf{L}_{j,j}\right]\right)^2}{8 \sum_{j=1}^{T/k} \tr\left[ \mathbf{L}_{j,j}^\T \mathrm{blkdiag}(\Delta^\T \Delta) \mathbf{L}_{j,j}\right]^2  } \right)
\end{aligned}
\end{equation}
by optimizing $\lambda$ in the last line. The point is that the bound \eqref{eq:simplechernoff} decays exponentially in $T$ as long as blocks on the diagonal of $\mathbf{L}$ have order constant condition number. In most applications, this can typically be achieved by a judicious choice of $k$. This leads us to define the following parameter:
\begin{equation}\label{eq:hypcon}
    \psi_k \triangleq \inf_{\Delta\in \R^{d'\times d}\setminus\{0\}} \left\{\frac{\left(\sum_{j=1}^{T/k} \tr\left[ \mathbf{L}_{j,j}^\T \mathrm{blkdiag}(\Delta^\T \Delta) \mathbf{L}_{j,j}\right]\right)^2}{ T\sum_{j=1}^{T/k} \tr\left[ \mathbf{L}_{j,j}^\T \mathrm{blkdiag}(\Delta^\T \Delta) \mathbf{L}_{j,j}\right]^2  } \right\}.
\end{equation}
Here, \eqref{eq:hypcon} is essentially a moment equivalence condition \citep[Cf.][Definition 4.1]{ziemann2022learning}. Note that $\psi_k$ depends implictly on $k$ since the block-length dictates the covariance structure of $\tilde X_{0:T-1}$. We remark that if all the diagonal blocks of $\mathbf{L}$ are identical, the process $\tilde X_{0:T-1}$ has period $k$. Hence in which case by Cauchy-Schwarz: $\psi_k \geq 1/k$. This is for instance true for any linear time invariant dynamics and thus, for these, we always have at least $\psi_k \geq 1/k$.   Returning to our over-arching goal of providing control of the smallest eigenvalue of the empirical covariance matrix \eqref{eq:empcov}, we now combine \eqref{eq:simplechernoff} (using $d'=1$) with a union bound.

\begin{theorem}\label{thm:anticonc}
Suppose $ \lambda_{\min} \left(\sum_{t=0}^{T-1} \E \tilde X_t \tilde X_t^\T \right)>0$. Under the hypotheses of \Cref{thm:expineq} we have that:
 \begin{multline}
    \mathbf{P} \left( \lambda_{\min} \left( \frac{1}{T} \sum_{t=0}^{T-1} X_t X_t^\T\right) \leq  \lambda_{\min}\left(\frac{1}{8T} \sum_{t=0}^{T-1}   \E \tilde X_t \tilde X_t^\T\right) \right)
    \\
    \leq   \Bigg(16
    \sqrt{
    1+\frac{\psi_k T \lambda_{\max} (\E [X_{0:T-1}X_{0:T-1}^\T] )}{ \lambda_{\min} \left(\sum_{t=0}^{T-1} \E [X_tX_t^\T] \right)}
    }
    %
    \sqrt{\frac{  \lambda_{\max} \left(\sum_{t=0}^{T-1} \E X_t  X_t\right) }{ \lambda_{\min} \left(\sum_{t=0}^{T-1} \E \tilde X_t \tilde X_t^\T \right)}}\Bigg)^d 
    \exp \left( \frac{-\psi_k T}{8} \right)
    .
 \end{multline}
\end{theorem}
Note that we always have---although this is far from sharp:\footnote{Write $\mathbf{L}$ in terms of $T$-many block rows to express $\mathbf{L}^\T\mathbf{L}$ as sums of products of these rows and then apply the triangle inequality. An improvement on this estimate is possible for instance if the process is a stable linear system, see \cite{jedra2022finite}.}
\begin{equation}
\frac{\lambda_{\max} (\E [X_{0:T-1}X_{0:T-1}^\T] )}
{
\lambda_{\min} \left(\sum_{t=0}^{T-1} \E [X_tX_t^\T] \right)}  \leq  
\frac{ \sum_{t=0}^{T-1}\lambda_{\max} \left( \E X_t X_t^\T \right) }
{\lambda_{\min} \left(\sum_{t=0}^{T-1} \E \tilde X_t \tilde  X_t^\T \right)
}.
\end{equation}
As long as $\sum_{t=0}^{T-1}\lambda_{\max} \left( \E X_t X_t^\T \right) \Big/\lambda_{\min} \left(\sum_{t=0}^{T-1} \E \tilde X_t \tilde  X_t^\T \right) = O(\mathrm{poly}(T))  $ and $\psi_k=O(T^\alpha)$ for some $\alpha \in (0,1)$, \Cref{thm:anticonc} gives a nontrivial lower bound on the smallest eigenvalue of \eqref{eq:empcov} which holds with probabilty approaching $1$ at an exponential rate in the sample size $T$.


\section{Example: Identification of Vector Autoregressions}\label{sec:arma}
We consider linear time-invariant dynamics of the form:
\begin{equation}\label{eq:ARMA}
    Z_{t} = \sum_{l=1}^{L}A_{l}Z_{t-l}+ HW_{t}, \qquad Z_{-L:-1}=0\qquad t=0,1,2,\dots
\end{equation}
where each $A_l\in \R^{d\times d}$ with $l \in [L]$ and $H \in \R^{d\times p}$. 

    Let  $
    \kappa \triangleq \left\{ \inf k : \det \left(\sum_{t=0}^{k-1} \E Z_{t:t-L+1} Z_{t:t-L+1}^\T\right) \neq 0 \right\}.$ Set also $\Gamma_k = \frac{1}{k} \sum_{t=0}^{k-1}\E Z_{t:t-L+1} Z_{t:t-L+1}^\T$ and   let us  define $A \in \R^{dL\times dL}$ by:
\begin{align}\label{eq:defAB}
    A &\triangleq \begin{bmatrix} 
    A_1 & A_2 &\dots  &  \dots & A_L\\
    I_d & 0 &\dots & \dots & 0\\
    0 & I_d & 0 & \dots & \vdots \\
    \vdots & \ddots & \ddots & \ddots &\vdots\\
    0 & \dots & 0 & I_d & 0
    \end{bmatrix}.
\end{align}

With these definitions in place, we may invoke \Cref{thm:anticonc} to control the empirical covariance of $X_t=Z_{t:t-L+1}$. We will subsequently use the lower bound of \Cref{thm:anticonc}  as an ingredient toward obtaining a non-asymptotic guarantee for least squares identification of vector autoregressions of order $L$.

\begin{corollary}\label{corr:ARMA}
Fix an integer $k \geq \kappa$ such that $T/k \in \N$. If $Z_{0:T-1}$ is given by \eqref{eq:ARMA}, we have that:
 \begin{multline}\label{eq:ARMAbound}
    \mathbf{P} \left( \lambda_{\min} \left( \frac{1}{T} \sum_{t=0}^{T-1} Z_{t:t-L+1} Z_{t:t-L+1}^\T\right) \leq \frac{1}{8} \lambda_{\min}\left(\Gamma_k\right) \right)
    \\
    \leq   \left(
    \frac{   32  T^{3/2}   \sum_{t=0}^{T-1}\bigopnorm{HH^\T}\bigopnorm{   A^{t-1} (A^{t-1})^\T} }{ \sqrt{k }\lambda_{\min} \left( \Gamma_k \right)}\right)^d 
    %
    \times \exp \left( \frac{- T}{8k} \right)
    .
\end{multline}
\end{corollary}
The proof of the above corollary follows immediately by \Cref{thm:anticonc}, \Cref{lem:armastability} combined with the observation that we may choose $\psi_k \geq 1/k$.

A few remarks are in order. First, \eqref{eq:ARMAbound} provides nontrivial control of the smallest eigenvalue of the empirical covariance of any ARMA process that satsifies: 1. the matrix $A$ in \eqref{eq:defAB} satisfies $\rho(A) \leq 1$ (marginal stability); and 2. $\kappa < \infty$ (controllability). The second condition can be further simplified if $\E (HW_t)(HW_t)^\T=HH^\T \succ 0$. Indeed, in this case, by observing that $A$ has downshift action, we see that an excitation of $\kappa = L$ is sufficient. Finally, we note that when specialized to first order processes, our result essentially recover \cite[Section D.1]{simchowitz2018learning}---our failure probabilities match with theirs up to logarithmic factors.

We now provide an identification guarantee for recovering the parameters $A_{1:L}\triangleq A_\star$. The argument rests on the decomposition \eqref{eq:LSEerror} and then combines \Cref{corr:ARMA} with a self-normalized martingale bound due to \cite{pena2009self,abbasi2011improved}.
\begin{theorem}
\label{thm:arconvergence}
Fix $\delta\in(0,1)$, an integer $k \geq \kappa$ such that $T/k \in \N$. Let $Z_{0:T-1}$ be given by \eqref{eq:ARMA} and suppose further that
\begin{equation}\label{eq:burnin}
    \frac{T}{8k} \geq  d\log \left(
    \frac{   32 T^{3/2}   \sum_{t=0}^{T-1}\bigopnorm{HH^\T}\bigopnorm{   A^{t-1} (A^{t-1})^\T} }{ \sqrt{k }\lambda_{\min} \left( \Gamma_k \right)}\right)
    +\log(1/\delta).
\end{equation}
It then holds on an event of probability at least $1-2\delta$ that the least squares estimator for $A_\star=A_{1:L}$ achieves:
\begin{equation*}
    \opnorm{\widehat A-A_\star} \leq \frac{32\opnorm{H}}{\sqrt{T\lambda_{\min}(\Gamma_k)}} 
    \sqrt{dL \log C_{\mathsf{SYS}}(T,k) +2d\log 5+ 2\log\frac{1}{\delta}}
\end{equation*}
where $C_{\mathsf{SYS}}(T,k) \triangleq 1+ \frac{ 32 \left(\sum_{t=1}^{T}\lambda_{\max} ( \E X_t X_t^\T)\right)^2}{ \left(\lambda_{\min} \left(\sum_{t=1}^{k} \E [X_tX_t^\T] \right)\right)^2}$.

\end{theorem}
We are thus able to recover the main result of \cite{simchowitz2018learning} and extend it to higher order lags ($L>1$) with slightly modified (logarithmic) dependencies on system parameters (and a slightly improved dependency on $\delta$).
\section{Proofs}

\paragraph{Proof of \Cref{thm:expineq}}
Let $\E_{T-k-1}$ denote conditioning with respect to $X_{0:T-k-1}$.
By repeated use of the tower property we have that:
\begin{equation}\label{eq:towerprop}
\begin{aligned}
      &\E \exp \left(-\lambda \sum_{t=0}^{T-1}\|\Delta X_t\|_2^2 \right) \leq   \E\exp \left(-\lambda \sum_{t=0}^{k-1}\|\Delta X_t\|_2^2\right) \times 
        \\ &
      \dots \times \E_{T-k-1} \exp \left(-\lambda \sum_{t=T-k}^{T-1}\|\Delta X_t\|_2^2 \right).
\end{aligned}
\end{equation}
We will bound each conditional expectation in \eqref{eq:towerprop} separately. Observe that
\begin{align*}
    &\sum_{t=T-k}^{T-1}  \|\Delta X_t\|_2^2  =  \begin{bmatrix}
    \Delta X_{T-k}\\
    \vdots\\
    \Delta X_{T-1}
    \end{bmatrix}^\T \begin{bmatrix}
    \Delta X_{T-k}\\
    \vdots\\
    \Delta X_{T-1}
    \end{bmatrix}
    \\&
    =  
    W_{T-k:T-1}^\T\mathbf{L}_{T/k}^\T \mathrm{blkdiag}(\Delta^\T \Delta) \mathbf{L}_{T/k} W_{T-k:T-1}
\end{align*}
In light of \Cref{lem:gausscondlem} we have that:
\begin{multline*}
   \E_{T-k-1} \exp \left(-\lambda \sum_{t=T-k}^{T-1}\|\Delta X_t\|_2^2\right) 
   \\
    \leq \exp \Bigg( -\lambda\tr\left[ \mathbf{L}_{T/k,T/k}^\T \mathrm{blkdiag}(\Delta^\T \Delta) \mathbf{L}_{T/k,T/k}\right] 
     \\
    + \frac{\lambda^2}{2}\tr\left[ \mathbf{L}_{T/k,T/k}^\T \mathrm{blkdiag}(\Delta^\T \Delta) \mathbf{L}_{T/k,T/k}\right]^2 \Bigg).
\end{multline*}
Repeatedly applying \Cref{lem:gausscondlem} as above yields the result. \hfill $\blacksquare$

\paragraph{Proof of \Cref{thm:anticonc}}
Let $\mathcal{N}_\e$ be an optimal $\e$-cover of the unit sphere $\mathbb{S}^{d-1}$. We begin with the following observation which is true for any $v \in \mathbb{S}^{d-1}$ and $v_i\in \mathcal{N}_\e$:
\begin{equation*}
    \begin{aligned}
       &\frac{1}{T} \sum_{t=0}^{T-1} v^\T X_tX_t^\T v \\
       &\geq \frac{1}{2T} \sum_{t=0}^{T-1} v^\T_i X_tX_t^\T v_i - \frac{1}{2T}\sum_{t=0}^{T-1} (v-v_i)^\T X_tX_t^\T (v-v_i).
    \end{aligned}
\end{equation*}
The rest of the proof consists of lower bounding the first term uniformly over $\mathcal{N}_\e$ and showing that the second term is of smaller order. To this end we now fix a multiplier $q\in (1,\infty)$. We define the events (i.e. $\Delta = v^\T$):
\begin{equation}
\begin{aligned}
    \mathcal{E}_1 &= \bigcup_{v\in \mathcal{N}_{\e}} \left\{ \frac{1}{T} \sum_{t=0}^{T-1} v^\T X_tX_t^\T v \leq \frac{1}{2T} \sum_{t=0}^{T-1} \E v^\T \tilde X_t\tilde X_t^\T v  \right\}
    \\
    \mathcal{E}_2&  =\Bigg\{\left\| \sum_{t=0}^{T-1}X_tX_t^\T\right\|_{\mathsf{op}}
    \geq 2q \times  \left\| \sum_{t=0}^{T-1}\E X_tX_t^\T\right\|_{\mathsf{op}} \Bigg\}.
\end{aligned}
\end{equation}
for any $v$, it is true on the complement of $\mathcal{E}= \mathcal{E}_1 \cup \mathcal{E}_2$ that for every $v_i \in \mathcal{N}_\e$:
\begin{equation*}
    \begin{aligned}
       &\frac{1}{T} \sum_{t=0}^{T-1} v^\T X_tX_t^\T v \\
       &\geq \frac{1}{2T} \sum_{t=0}^{T-1} v^\T_i X_tX_t^\T v_i - \frac{1}{2T}\sum_{t=0}^{T-1} (v-v_i)^\T X_tX_t^\T (v-v_i)  \\
       &\geq \frac{1}{4T}\sum_{t=0}^{T-1} \E  v_i^\T \tilde X_t  \tilde  X_t^\T v_i-\frac{q\e^2}{T}\left\| \sum_{t=0}^{T-1}\E X_tX_t^\T\right\|_{\mathsf{op}} 
    \end{aligned}
\end{equation*}
where $v-v_i$ has norm at most $\e$ for some choice of $v_i$ by the covering property. For this choice we have that:
\begin{align*}
      \frac{1}{T} \sum_{t=0}^{T-1} v^\T X_tX_t^\T v
       &\geq \frac{1}{8T}\sum_{t=0}^{T-1}  v_i^\T \E [\tilde X_t \tilde X_t^\T] v_i
\end{align*}
as long as:
\begin{align*}
    \e^2 \leq \frac{ \lambda_{\min} \left(\sum_{t=0}^{T-1} \E \tilde X_t \tilde X_t^\T \right)}{8q  \lambda_{\max} \left(\sum_{t=0}^{T-1} \E X_t  X_t\right) }.
\end{align*}
To finish the proof, it suffices to estimate the failure probabilities $\mathbf{P}(\mathcal{E}_1)$ and $\mathbf{P}(\mathcal{E}_2)$. By \eqref{eq:simplechernoff}, a volumetric argument \citep[see e.g.][Example 5.8]{wainwright2019high} (which controls the cardinality of $\mathcal{N}_\e$) and our particular choice of $\e$ we have:
\begin{align*}
    \mathbf{P}(\mathcal{E}_1)& \leq \left(1+\frac{2}{\e^2} \right)^d  \exp \left( \frac{-\psi_k T}{8} \right)\\
    &\leq  \left(8\sqrt{\frac{  q\lambda_{\max} \left(\sum_{t=0}^{T-1} \E X_t  X_t\right) }{ \lambda_{\min} \left(\sum_{t=0}^{T-1} \E \tilde X_t \tilde X_t^\T \right)}}\right)^d \exp \left( \frac{-\psi_k T}{8} \right).
\end{align*}
The event $\mathcal{E}_2$ is controlled by \eqref{eq:theuppertail} which yields:
\begin{align*}
    \mathbf{P}(\mathcal{E}_2) \leq   5^d
   \exp \left( \frac{-(q-1) \lambda_{\min} \left(\sum_{t=0}^{T-1}  \E[X_t X_t^\T]\right) }{8\lambda_{\max} (\mathbf{L}^\T \mathbf{L} )} \right).
\end{align*}
By choosing $$q=1+\frac{\psi_k T \lambda_{\max} (\mathbf{L}^\T \mathbf{L} )}{\lambda_{\min} \left(\sum_{t=0}^{T-1} \E [X_tX_t^\T] \right)},$$
the result holds on the complement of $\mathcal{E}_1\cup \mathcal{E}_2$ and thus also holds with the desired probability. \hfill $\blacksquare$

\subsection{Proofs related to AR processes}

\begin{lemma}\label{lem:armastability}
For $Z_{0:T-1}$ given by \eqref{eq:ARMA} and $X_t = Z_{t:t-L+1}$ we have that:
\begin{align*}
    \bigopnorm{\sum_{t=0}^{T-1}\E X_t X_t^T } 
    \leq   T \opnorm{HH^\T}  \sum_{k=0}^{T-1}\bigopnorm{   A^{T-k-1} (A^{T-k-1})^\T}
\end{align*}
\end{lemma}

\begin{proof}
We have that $X_{t+1} = AX_t + BW_t$  where $B=\begin{bmatrix}
        H & 0 &\dots & 0
    \end{bmatrix}^\T$.  Notice now that $
    X_t = \sum_{k=0}^{t-1} A^{t-k-1} BW_{k}$. It is straightforward to verify that  for $t\in [T]$:
\begin{align*}
    \E X_t X_t^\T  =  \sum_{k=0}^{t-1}   A^{t-k-1} B \E [W_{k} W_k^\T] B^\T (A^{t-k-1})^\T.
\end{align*}
Since each $W_k$ has identity covariance, we thus also have that:
\begin{equation}
\begin{aligned}
    \sum_{t=0}^{T-1}\bigopnorm{  \E  X_t X_t^\T} \leq \sum_{t=0}^{T-1} \bigopnorm{ \sum_{k=0}^{t-1}   A^{t-k-1} BB^\T (A^{t-k-1})^\T}\\
     \leq T \opnorm{BB^\T}  \sum_{k=0}^{T-1}\bigopnorm{   A^{T-k-1} (A^{T-k-1})^\T}.
\end{aligned}
\end{equation}
The result follows by noticing that $\opnorm{BB^\T} =\bigopnorm{HH^\T}$.
\end{proof}

\paragraph{Proof of \Cref{thm:arconvergence}}
Let $V_t = HW_t$ and note that this is $\opnorm{H}^2$-sub-Gaussian. If we combine \eqref{eq:LSEerror} with \Cref{corr:ARMA} we find that as long as \eqref{eq:burnin} holds we have that with probability at least $1-\delta$: 
\begin{align*}
    \opnorm{\widehat A-A_\star} \leq \frac{16}{\sqrt{T\lambda_{\min}(\Gamma_k)}}
    \bigopnorm{\left(\sum_{t=0}^{T-1} V_t X_t^\T \right)\left(\sum_{t=0}^{T-1} X_t X_t^\T+\frac{T}{16}\Gamma_k \right)^{-1/2 }} .
\end{align*}
Let now $v_{1:5^{d}}$ be an $\e$-net of the $d$-dimensional unit sphere with $\e=0.5$. Such a net exists by virtue of a standard volumetric argument \cite[see e.g][Example 5.8]{wainwright2019high}. Discretizing the operator norm yields:
\begin{multline*}
    \bigopnorm{\left(\sum_{t=1}^T V_t X_t^\T \right)\left(\frac{T}{16}\Gamma_k+ \sum_{t=1}^T X_t X_t^\T \right)^{-1/2 }}^2
    \\
        \leq  2 \sup_{i\in[5^d]} \bignorm{v^\T_i \left(\sum_{t=1}^T V_t X_t^\T \right)\left(\frac{T}{16}\Gamma_k+ \sum_{t=1}^T X_t X_t^\T \right)^{-1/2 }}^2.
\end{multline*}
If we combine Theorem 1 of \cite{abbasi2011improved} with a union bound over the $5^{d}$ elements above we arrive at that with probability at least $1- \delta$:
    \begin{multline*}
        2\sup_{i\in [5^d]} \bignorm{v_i^\T \left(\sum_{t=1}^T V_t X_t^\T \right)\left(\frac{T}{16}\Gamma_k + \sum_{t=1}^T X_t X_t^\T \right)^{-1/2 }} \\
        \leq \Bigg(4 \sigma^2 \log\paren{\det\paren{I+ \frac{16}{T}\sum_{t=1}^T X_t X_t^\T \Gamma_k^{-1} } }
        +8d\sigma^2\log 5+ 8\sigma^2\log\frac{1}{\delta}\Bigg)^{1/2}
    \end{multline*} 
    where $\sigma=\opnorm{H}$.

To finish the proof, it remains to control $ \sum_{t=1}^T X_t X_t^\T$. However, part of the proof of \Cref{thm:anticonc} actually reveals that on the same event as above we have that 
\begin{equation*}
    \bigopnorm{ \sum_{t=1}^T X_t X_t^\T } \leq\frac{ 2T \left(\sum_{t=1}^{T}\lambda_{\max} ( \E X_t X_t^\T)\right)^2}{k\lambda_{\min} \left(\sum_{t=1}^{T} \E [X_tX_t^\T] \right)}.
\end{equation*}
Hence
\begin{equation*}
     \bigopnorm{ \frac{16}{T}\sum_{t=1}^T X_t X_t^\T \Gamma_k^{-1}} \leq \frac{ 32 \left(\sum_{t=1}^{T}\lambda_{\max} ( \E X_t X_t^\T)\right)^2}{ \left(\lambda_{\min} \left(\sum_{t=1}^{k} \E [X_tX_t^\T] \right)\right)^2}
\end{equation*}
and the result follows by bounding the determinant above by $C_{\mathsf{SYS}}$ (an upper bound on the relevant largest eigenvalue) raised to the power of its dimension---$dL$.
\hfill$\blacksquare$

\subsection{Facts about the Gaussian distribution}

We begin by stating a version of Lemma 2.1 in \cite{tu2023elementary}. To make this note self-contained, we provide a short proof.

\begin{lemma}
    Fix $x\in \R^n$ and let $W\sim N(0,I_m)$. For any positive semidefinite $Q \in \R^{(n+m)\times(n+m)}$ of the form $ Q=\begin{bmatrix}Q_{11}& Q_{12}\\ Q_{21} & Q_{22} \end{bmatrix}$ and any $\lambda \geq 0$ we have that:
    \begin{equation}
    \begin{aligned}\label{eq:stephenslemma}
        \E \exp \left( -\lambda \begin{bmatrix}x \\ W
    \end{bmatrix}^\T  \begin{bmatrix}Q_{11}& Q_{12}\\ Q_{21} & Q_{22} \end{bmatrix} \begin{bmatrix}x \\ W
    \end{bmatrix}\right) 
    \leq \left(\det (I+2\lambda Q_{22})\right)^{-1/2}.
    \end{aligned} 
    \end{equation}
\end{lemma}

\begin{proof}
Let $Q_\lambda  \triangleq \begin{bmatrix}Q_{11}& Q_{12}\\ Q_{21} & Q_{22}-(2\lambda )^{-1} I_m \end{bmatrix} $. We then have: 
\begin{multline}\label{eq:proofofstephenslemma}
     \E \exp \left( -\lambda \begin{bmatrix}x \\ W
    \end{bmatrix}^\T  \begin{bmatrix}Q_{11}& Q_{12}\\ Q_{21} & Q_{22} \end{bmatrix} \begin{bmatrix}x \\ W
    \end{bmatrix}\right) \\
    =  \int_{\R^m}\exp \left( -\lambda \begin{bmatrix}x \\ w
    \end{bmatrix}^\T  \begin{bmatrix}Q_{11}& Q_{12}\\ Q_{21} & Q_{22}-(2\lambda )^{-1} I_m \end{bmatrix} \begin{bmatrix}x \\ w
    \end{bmatrix}\right) dw\\
     =  \exp\left(-\lambda  x^\T   (Q_\lambda/Q_{22})x\right)\\
     \times \int_{\R^m}\exp \left(  (w+\mu)^\T ( Q_{22}-(2\lambda )^{-1} I_m )(w+\mu) \right) dw
\end{multline}
by using the LDU decomposition of $Q_\lambda$ to block-diagonalize and where $\mu = (Q_{22}+(2\lambda)^{-1})^{-1}Q_{12} x$. Since $(Q_\lambda/Q_{22}) \succeq 0$, we have that $\exp\left(-\lambda  x^\T   (Q_\lambda/Q_{22})x\right) \leq 1$ and it is readily verified that the integral on the last line of  \eqref{eq:proofofstephenslemma} evaluates to right hand side of \eqref{eq:stephenslemma}, as per requirement.
\end{proof}

\gausscondlem*

\paragraph{Proof of \Cref{lem:gausscondlem}}
We take \eqref{eq:stephenslemma} as a starting point and manipulate the determinant on the right hand side. In particular, by writing the determinant as a sum ($\exp \circ \log =\mathrm{identity}$) and by  invoking $\log(1+x) \geq  x- x^2/2$ (valid for $x\geq 0$) for each eigenvalue, we have the result. \hfill $\blacksquare$

\begin{lemma}
For any $\lambda \in \left[0, \frac{1}{4\lambda_{\max}(\mathbf{L}^\T \mathbf{L})}\right]$ and $v\in \R^d$ with $\|v\|_2^2\leq 1$, we have that:
\begin{align*}
    \E \exp \left(\lambda \sum_{t=0}^{T-1} v^\T X_t X_t^\T v \right) \leq \exp \left( 4\lambda \sum_{t=0}^{T-1} v^\T \E [X_t X_t^\T] v \right).
\end{align*}

\begin{proof}
Let $\mathbf{L}_v = ( I_{T}\otimes v^\T )\mathbf{L}$. Since $(v^\T X)_{0:T-1} =\mathbf{L}_v W_{0:T-1}$, a standard calculation gives
\begin{align*}
     &
     \E \exp \left(\lambda W_{0:T-1}^\T \mathbf{L}_v^\T \mathbf{L}_v W_{0:T-1} \right) =
     \left( \det(I-2\lambda \mathbf{L}_v^\T \mathbf{L}_v ) \right)^{-1/2}
     \\ &
     = \exp \left( -\sum_{i=1}^{Td} \log \left(1-2\lambda \times \lambda_i(\mathbf{L}_v^\T \mathbf{L}_v )\right)\right).
\end{align*}
The result follows by repeated application of the numerical inequality: $ -\log(1-x) \leq 2x $ (which is valid for all $x\in [0, 1/2]$). 
\end{proof}

\end{lemma}

The preceding lemma easily yields an upper tail-bound for the empirical covariance by a Chernoff argument:
\begin{equation*}
    \begin{aligned}
    \mathbf{P} \left(  \sum_{t=0}^{T-1} v^\T X_t X_t^\T v  \geq  q\sum_{t=0}^{T-1} v^\T \E [ X_t X_t^\T] v \right)
    \leq \exp \left( \frac{-(q-1) \sum_{t=0}^{T-1} v^\T \E[X_t X_t^\T]v }{8\lambda_{\max} (\mathbf{L}^\T \mathbf{L} )} \right).
    \end{aligned}
\end{equation*}

In turn, combining \eqref{eq:theuppertail} with an $\e$-net argument and a union bound we arrive at the following \citep[cf.][Exercise 4.4.3b]{vershynin2018}.

\begin{multline}\label{eq:theuppertail}
    \Pr\left( \bigopnorm{\frac{1}{T} \sum_{t=0}^{T-1}X_tX_t^\T} \geq 2q  \bigopnorm{\frac{1}{T} \sum_{t=0}^{T-1}\E X_tX_t^\T}\right)\\ \leq 5^d  \exp \left( \frac{-(q-1)   \lambda_{\min} \left(\sum_{t=0}^{T-1}\E X_tX_t^\T \right) }{8\lambda_{\max} (\mathbf{L}^\T \mathbf{L} )} \right).
\end{multline}

\paragraph{Acknowledgements} This note was prepared while the author was still at KTH. It was prompted by a question asked by Samet Oymak (can the method described in \cite{ziemann2022learning} give sharp bounds for linear Gaussian models?). The author also thanks Yassir Jedra, Henrik Sandberg and Anastasios Tsiamis for several helpful discussions and acknowledges support by the Swedish Research Council (grant 2016-00861).

\bibliographystyle{alpha}
\bibliography{main.bib}

\newcommand{\etalchar}[1]{$^{#1}$}
\begin{thebibliography}{WRMDM05}

\bibitem[AYPS11]{abbasi2011improved}
Yasin Abbasi-Yadkori, D{\'a}vid P{\'a}l, and Csaba Szepesv{\'a}ri.
\newblock Improved algorithms for linear stochastic bandits.
\newblock {\em Advances in neural information processing systems}, 24, 2011.

\bibitem[FTM18]{faradonbeh2018finite}
Mohamad Kazem~Shirani Faradonbeh, Ambuj Tewari, and George Michailidis.
\newblock Finite time identification in unstable linear systems.
\newblock {\em Automatica}, 96:342--353, 2018.

\bibitem[JP22]{jedra2022finite}
Yassir Jedra and Alexandre Proutiere.
\newblock Finite-time identification of linear systems: Fundamental limits and
  optimal algorithms.
\newblock {\em IEEE Transactions on Automatic Control}, 2022.

\bibitem[Lju99]{lennart1999system}
Lennart Ljung.
\newblock System identification: theory for the user.
\newblock {\em PTR Prentice Hall, Upper Saddle River, NJ}, 28, 1999.

\bibitem[Men14]{mendelson2014learning}
Shahar Mendelson.
\newblock Learning without concentration.
\newblock In {\em Conference on Learning Theory}, pages 25--39. PMLR, 2014.

\bibitem[OO21]{oymak2021revisiting}
Samet Oymak and Necmiye Ozay.
\newblock Revisiting ho--kalman-based system identification: Robustness and
  finite-sample analysis.
\newblock {\em IEEE Transactions on Automatic Control}, 67(4):1914--1928, 2021.

\bibitem[PLS09]{pena2009self}
Victor~H Pe{\~n}a, Tze~Leung Lai, and Qi-Man Shao.
\newblock {\em Self-normalized processes: Limit theory and Statistical
  Applications}.
\newblock Springer, 2009.

\bibitem[SMT{\etalchar{+}}18]{simchowitz2018learning}
Max Simchowitz, Horia Mania, Stephen Tu, Michael~I. Jordan, and Benjamin Recht.
\newblock Learning without mixing: Towards a sharp analysis of linear system
  identification.
\newblock In {\em Conference On Learning Theory}, pages 439--473. PMLR, 2018.

\bibitem[SR19]{sarkar2019near}
Tuhin Sarkar and Alexander Rakhlin.
\newblock {Near Optimal Finite Time Identification of Arbitrary Linear
  Dynamical Systems}.
\newblock In {\em International Conference on Machine Learning}, pages
  5610--5618, 2019.

\bibitem[TB23]{tu2023elementary}
Stephen Tu and Ross Boczar.
\newblock An elementary proof of anti-concentration for degree two non-negative
  gaussian polynomials.
\newblock {\em arXiv preprint arXiv:2301.05992}, 2023.

\bibitem[TFS22]{tu2022learning}
Stephen Tu, Roy Frostig, and Mahdi Soltanolkotabi.
\newblock Learning from many trajectories.
\newblock {\em arXiv preprint arXiv:2203.17193}, 2022.

\bibitem[TP19]{tsiamis2019finite}
Anastasios Tsiamis and George~J. Pappas.
\newblock Finite sample analysis of stochastic system identification.
\newblock In {\em 2019 IEEE 58th Conference on Decision and Control (CDC)},
  pages 3648--3654. IEEE, 2019.

\bibitem[TZMP22]{tsiamis2022statistical}
Anastasios Tsiamis, Ingvar Ziemann, Nikolai Matni, and George~J Pappas.
\newblock Statistical learning theory for control: A finite sample perspective.
\newblock {\em arXiv preprint arXiv:2209.05423}, 2022.

\bibitem[Ver18]{vershynin2018}
Roman Vershynin.
\newblock {\em High-Dimensional Probability: An Introduction with Applications
  in Data Science}.
\newblock Cambridge University Press, 2018.

\bibitem[Wai19]{wainwright2019high}
Martin~J. Wainwright.
\newblock {\em High-dimensional statistics: A non-asymptotic viewpoint},
  volume~48.
\newblock Cambridge University Press, 2019.

\bibitem[WRMDM05]{willems2005note}
Jan~C Willems, Paolo Rapisarda, Ivan Markovsky, and Bart~LM De~Moor.
\newblock {A Note on Persistency of Excitation}.
\newblock {\em Systems \& Control Letters}, 54(4):325--329, 2005.

\bibitem[ZT22]{ziemann2022learning}
Ingvar Ziemann and Stephen Tu.
\newblock Learning with little mixing.
\newblock {\em arXiv preprint arXiv:2206.08269. {NeurIPS'22}}, 2022.

\end{thebibliography}

\end{document}